\documentclass[11pt]{article}

\usepackage[svgnames]{xcolor} 
\usepackage{fullpage}
\usepackage{tikz-cd}
\usepackage{etoolbox}
\usepackage{amsmath,amssymb,amsthm}
\usepackage{microtype}

\usepackage{chngcntr} 
\usepackage{apptools}
\AtAppendix{\counterwithin{theorem}{subsection}} 
\usepackage{thmtools}
\usepackage{thm-restate} 
\usepackage{tikz}

\newtheorem{theorem}{Theorem}
\newtheorem{lemma}[theorem]{Lemma}
\newtheorem{corollary}[theorem]{Corollary}
\newtheorem{definition}[theorem]{Definition}

\newtheorem{remark}{Remark}

\usepackage[paper,hyperref]{knowledge}
\knowledgeconfigure{notion,quotation,protect quotation={tikzcd}}
\renewrobustcmd\phantomintro[1]{\kl[#1]{}}

\usetikzlibrary{decorations.markings} \tikzset{negated/.style={
    decoration={markings, mark= at position 0.5 with { \node[transform
        shape,xscale=.8,yscale=.4] (tempnode) {$\slash$}; } },
    postaction={decorate} } }

\usepackage{cite}

\knowledgedirective{math notion}{notion}

\knowledge {category representing input words}{}

\knowledge {knowledged}{url={https://ctan.org/pkg/knowledge?lang=en},italic}

\knowledge {dfaclassic}{notion}

\knowledge {statesclassic}{notion}
\knowledge {initialclassic}{notion}
\knowledge {finalclassic}{notion}
\knowledge {transition classic}{notion}
\knowledge {runclassic}{notion}
\knowledge {acceptingclassic}{notion}
\knowledge {acceptsclassic}{synonym}
\knowledge {languageclassic}{notion}
\knowledge {non-deterministic automaton}{notion}

\knowledge {$\catC $-automaton}{notion}
\knowledge {$\catC $-automata}{synonym}
\knowledge {$\catD $-automaton}{synonym}
\knowledge {automata in the category}{synonym}
\knowledge {automaton in $\KlTrans $}{synonym}
\knowledge {automaton in $\Set $}{synonym}
\knowledge {automata}{synonym}
\knowledge {automata in the Kleisli category of a suitably chosen monad}{synonym}
\knowledge {automaton}{synonym}
\knowledge {$\Rel $-automata}{synonym}
\knowledge {$\Rel $-automaton}{synonym}
\knowledge {$\Setop$-automaton}{synonym}
\knowledge {$\Set$-automaton}{synonym}

\knowledge {$(\KlT , 1,1)$-automaton}{notion}
\knowledge {$(\Set , 1,B^*+1)$-automaton}{notion}

\knowledge {$\KlTrans $-automaton}{notion}
\knowledge {\Im }{notion}

\knowledge {$\catC $-language}{notion}
\knowledge {$\catC $-languages}{synonym}
\knowledge {languages}{synonym}
\knowledge {language}{synonym}

\knowledge {\Auto }{notion}

\knowledge {accept}[accepts]{notion}
\knowledge {$\KlTrans $-automata for $\langLKlT $}{synonym}
\knowledge {$\Set $-automata accepting $\langLSet $}{synonym}
\knowledge {for a language}{synonym}
\knowledge {language accepted}{synonym}
\knowledge {automata for a language}{synonym}
\knowledge {automata for languages}{synonym}
\knowledge {accepting}{synonym}

\knowledge {initial automaton}{notion}
\knowledge {initial $\langL $-automaton}{synonym}

\knowledge {final automaton}{notion}
\knowledge {final automata}{synonym}
\knowledge {final $\langL $-automaton}{synonym}

\knowledge {\objectLeft }{math notion}
\knowledge {\objectRight }{math notion}
\knowledge {\symbLeft }{math notion}
\knowledge {\symbRight }{math notion}
\knowledge {\objectCenter }{math notion}

\knowledge {$(\catC , X,Y)$-automaton}{notion}
\knowledge {$(\KlTrans ,1,1)$-automaton}{synonym}
\knowledge {$(\KlTrans ,1,1)$-automata}{synonym}
\knowledge {$(\Set ,1,\alphabetB ^*)$-automata}{synonym}

\knowledge {$(\catC , X,Y)$-language}{notion}
\knowledge {$(\catC ,I,GO)$-languages}{synonym}
\knowledge {$(\catD ,FI,O)$-languages}{synonym}
\knowledge {$(\catC ,I,GO)$-language}{synonym}
\knowledge {$(\catD ,FI,O)$-language}{synonym}
\knowledge {$(\Set ,1,2)$-language}{synonym}
\knowledge {$(\Setop ,2,1)$-language}{synonym}
\knowledge {$(\Rel ,1,1)$-language}{synonym}
\knowledge {$(\KlT , 1,1)$-language}{synonym}
\knowledge {$(\Set , 1,B^*+1)$-language}{synonym}
\knowledge {$(\KlT , \FreeTrans 1,1)$-language}{synonym}
\knowledge {$(\KlTrans ,1,1)$-language}{synonym}
\knowledge {$(\KlTrans ,1,1)$-languages}{synonym}
\knowledge {$(\Set , 1,\UTrans 1)$-language}{synonym}

\knowledge\catReach{math notion}
\knowledge\catObs{math notion}

\knowledge {\Min }{math notion}
\knowledge {\Obs }{math notion}
\knowledge {\Reach }{math notion}

\knowledge {-divides}{notion}
\knowledge {-divisibility}{synonym}

\knowledge {-quotient}{notion}
\knowledge {-quotients}{synonym}

\knowledge {-subobject}{notion}
\knowledge {-subobjects}{synonym}

\knowledge {-minimal}{notion}
\knowledge {minimal automaton}{synonym}
\knowledge {minimal}{synonym}

\knowledge \EpiAut{math notion}
\knowledge \MonoAut{math notion}

\knowledge {observable}{notion}
\knowledge {observable quotient automaton}{synonym}

\knowledge {reachable}{notion}
\knowledge {reachable sub-automaton}{synonym}

\knowledge {subsequential transducer}{notion}
\knowledge {Subsequential transducers}{synonym}
\knowledge {subsequential transducers}{synonym}

\knowledge {input alphabetSS}{notion}
\knowledge {output alphabetSS}{notion}
\knowledge {statesSS}{notion}
\knowledge {initial state of the transducer}{notion}
\knowledge {initialization value}{notion}

\knowledge {initialSS}{notion}

\knowledge {\semTrans }{math notion}

\knowledge {partial transition function for the letter~$a$}{notion}
\knowledge {partial transition function}{synonym}

\knowledge {partial termination function}{notion}

\knowledge {partial production function for the letter~$a$}{notion}
\knowledge {partial production function}{synonym}

\knowledge {initial transducer}{notion}
\knowledge {final transducer}{notion}
\knowledge {factorization transducer}{notion}

\knowledge {\monadTrans }{math notion}
\knowledge {\KlTrans }{math notion}
\knowledge {Kleisli category}{notion}
\knowledge {Kleisli category for $\monadTrans $}{synonym}

\knowledge {Kleisli morphism}{notion}
\knowledge {Kleisli morphisms}{synonym}

\knowledge {\FreeTrans}{math notion}
\knowledge {\UTrans}{math notion}

\knowledge {\botfun }{math notion}

\knowledge {longest common prefix}{notion}
\knowledge {\lcp }{math notion}

\knowledge {irreducible function}{notion}
\knowledge {irreducible}{synonym}
\knowledge {\Irr }{math notion}

\knowledge {reduction of~$K$}{notion}
\knowledge {\red }{math notion}

\knowledge {\suff }{math notion}
\knowledge {\pstar}{math notion}

\knowledge {Brzozowski's algorithm}{notion}
\knowledge {Brzozowski's minimization algorithm}{synonym}

\knowledge {\FPow }{math notion}
\knowledge {\UPow }{math notion}
\knowledge {\UPowop }{math notion}
\knowledge {\FPowop }{math notion}

\knowledge {\EpiKlT }{math notion}
\knowledge {\MonoKlT }{math notion}

\knowledge {\determinize }{math notion}
\knowledge {\transpose }{math notion}
\knowledge {\codeterminize }{math notion}
\knowledge {determinization functor}{notion}
\knowledge {codeterminization functor}{notion}

\knowledge {\inv }{}

\knowledge {category}{}
\knowledge {isomorphic}{}
\knowledge {object}{}
\knowledge {objects}{synonym}

\knowledge {initial object}{}
\knowledge {initial}{synonym}
\knowledge {final object}{}
\knowledge {final}{synonym}

\knowledge {isomorphism}{}
\knowledge {factorization}{}
\knowledge {factorization systems}{}
\knowledge {factorization system}{synonym}

\knowledge {adjunctions}{}
\knowledge {-factorization}{}
\knowledge {diagonal property}{}
\knowledge {Kan extensions}{}

\knowledge {regular}{}
\knowledge {regular epis}{}
\knowledge {algebraic categories}{}
\knowledge {quotients of algebras}{}
\knowledge {subalgebras}{}
\knowledge {composition}{}

\knowledge {knowledge}{url={https://ctan.org/pkg/knowledge?lang=en}}

\newrobustcmd\cat[1]{\mathcal{#1}}
\newrobustcmd\catA{\cat A}
\newrobustcmd\catC{\cat C}
\newrobustcmd\catD{\cat D}
\newrobustcmd\catE{\cat E}
\newrobustcmd\catI{\cat I}
\newrobustcmd\catO{\cat O}
\newrobustcmd\catS{\cat S}
\newrobustcmd\catK{\cat K}

\newrobustcmd\op{\mathit{op}}
\newrobustcmd\Auto{\kl[\Auto]{\mathsf{Auto}}}
\newrobustcmd\catAutoL{\Auto(\langL)}
\newrobustcmd\catAutoLC{\Auto(\langL_\catC)}
\newrobustcmd\catAutoLD{\Auto(\langL_\catD)}
\newrobustcmd\catAutoLSet{\Auto(\langL_\Set)}
\newrobustcmd\catAutoLSetop{\Auto(\langL_{\Setop})}
\newrobustcmd\catAutoLRel{\Auto(\langL_\Rel)}
\newrobustcmd\catAutoLRelop{\Auto(\langL_\Relop)}
\newrobustcmd\catAutoLRelrev{\Auto(\langL_\Rel^\inv)}
\newrobustcmd\catAutoLSetrev{\Auto(\langL_\Set^\inv)}
\newrobustcmd\Set{\mathsf{Set}}
\newrobustcmd\Setop{\mathsf{Set^\op}}
\newrobustcmd\Rel{\mathsf{Rel}}
\newrobustcmd\Relop{\mathsf{Rel^\op}}

\newrobustcmd\lang[1]{\mathcal{#1}}
\newrobustcmd\langL{\lang L}
\newrobustcmd\langLSet{\lang L_{\Set}}
\newrobustcmd\langLEMT{\lang L_{\EMT}}
\newrobustcmd\langLSetop{\lang L_{\Setop}}
\newrobustcmd\langLRel{\lang L_{\Rel}}
\newrobustcmd\langLRelop{\lang L_{\Relop}}
\newrobustcmd\langLRelrev{\lang L_{\Rel}^{\inv}}
\newrobustcmd\langLSetrev{\lang L_{\Set}^{\inv}}
\newrobustcmd\aut[1]{\mathcal{#1}}
\newrobustcmd\autA{\aut A}
\newrobustcmd\autB{\aut B}
\newrobustcmd\autC{\aut C}

\newrobustcmd\Reach{\kl[\Reach]{\mathtt{Reach}}}
\newrobustcmd\Obs{\kl[\Obs]{\mathtt{Obs}}}
\newrobustcmd\Min{\kl[\Min]{\mathtt{Min}}}
\newrobustcmd{\ReachA}{\Reach{(\autA)}}
\newrobustcmd{\ObsA}{\Obs{(\autA)}}
\newrobustcmd\MinL{\Min{(\langL)}}

\newrobustcmd\AutInit{\autA^{\kl[initial automaton]{\mathit{init}}}}
\newrobustcmd\AutFin{\autA^{\kl[final automaton]{\mathit{final}}}}

\newrobustcmd\autAinitL{\AutInit(\langL)}
\newrobustcmd\autAfinalL{\AutFin(\langL)}

\newrobustcmd\symbLeft{\kl[\symbLeft]\triangleright}
\newrobustcmd\symbRight{\kl[\symbRight]\triangleleft}
\newrobustcmd\objectLeft{\kl[\objectLeft]{\mathsf{in}}}
\newrobustcmd\objectRight{\kl[\objectRight]{\mathsf{out}}}
\newrobustcmd\objectCenter{\kl[\objectCenter]{\mathsf{states}}}

\newrobustcmd\EpiAutL{\kl[\EpiAutL]{\mathcal{E}_{\mathsf{Auto}(\langL)}}}
\newrobustcmd\MonoAutL{\kl[\MonoAutL]{\mathcal{M}_{\mathsf{Auto}(\langL)}}}

\newrobustcmd\EpiAut{\kl[\EpiAutL]{\mathcal{E}_{\mathsf{Auto}(\langL_{\KlTrans})}}}
\newrobustcmd\MonoAut{\kl[\MonoAutL]{\mathcal{M}_{\mathsf{Auto}(\langL_{\KlTrans})}}}
\knowledge{\EpiAutL}{math notion}
\knowledge{\MonoAutL}{math notion}

\newrobustcmd\Epi{\ensuremath{\mathcal{E}}}
\newrobustcmd\Mono{\ensuremath{\mathcal{M}}}

\newrobustcmd\Lan[2]{\mathsf{Lan}_{#2}{#1}}
\newrobustcmd\Ran[2]{\mathsf{Ran}_{#2}{#1}}

\newrobustcmd\St{\mathsf{State}}
\newrobustcmd\MB{M(B)}
\newrobustcmd\tokl\nrightarrow
\newrobustcmd\TM{T_{M}}
\newrobustcmd\conv[1]{\breve{#1}}

\newrobustcmd\catAutoLKlT{\Auto(\langL_{\KlTrans})}
\newrobustcmd\catAutoLEMT{\Auto(\langL_{\EMT})}
\newrobustcmd\monadTrans{\kl[\monadTrans]{\mathcal{T}}}
\newrobustcmd\KlTrans{\kl[\KlTrans]{\mathsf{Kl}(\mathcal{T})}}
\newrobustcmd\EMT{\kl[\EMT]{\mathsf{EM}(\mathcal{T})}}
\newrobustcmd\langLKlT{\lang L_{\KlTrans}}
\newrobustcmd\EpiKlT{\kl[\EpiKlT]{\mathcal{E}_{\mathsf{Kl}(\mathcal{T})}}}
\newrobustcmd\MonoKlT{\kl[\MonoKlT]{\mathcal{M}_{\mathsf{Kl}(\mathcal{T})}}}
\newrobustcmd\semTrans[1]{\kl[\semTrans]{[\![}#1\kl[\semTrans]{]\!]}}
\newrobustcmd\KlT{\KlTrans}
\newrobustcmd\FreeTrans{\kl[\FreeTrans]{F_{\mathcal{T}}}}
\newrobustcmd\FreeTransEM{\kl[\FreeTransEM]{F^{\mathcal{T}}}}
\newrobustcmd\UTrans{\kl[\UTrans]{U_{\mathcal{T}}}}
\newrobustcmd\UTransEM{\kl[\UTransEM]{U^{\mathcal{T}}}}

\newrobustcmd\pstar{\mathrel{\kl[\pstar]{\star}}}

\newrobustcmd\FreeAut{\overline{\FreeTrans}}
\newrobustcmd\FreeAutEM{\overline{\FreeTransEM}}
\newrobustcmd\AutUTrans{\overline{\UTrans}}
\newrobustcmd\AutUTransEM{\overline{\UTransEM}}
\newrobustcmd\dom{\mathrm{dom}}
\newrobustcmd\FPow{\kl[\FPow]{F_{\mathcal{P}}}}
\newrobustcmd\UPow{\kl[\UPow]{U_{\mathcal{P}}}}
\newrobustcmd\Pow{\mathcal{P}}
\newrobustcmd\rev{\mathcal{R}}
\newrobustcmd\FPowop{\kl[\FPowop]{F^\op_{\mathcal{P}}}}
\newrobustcmd\UPowop{\kl[\UPowop]{U^\op_{\mathcal{P}}}}

\newrobustcmd\botfun{\kl[\botfun]{\kappa_\bot}}
\newrobustcmd\IrrAB{\kl[\Irr]{\mathsf{Irr}}(A^*,B^*)}
\newrobustcmd\red{\kl[\red]{\mathsf{red}}}
\newrobustcmd\lcp{\kl[\lcp]{\mathsf{lcp}}}
\newrobustcmd\redL{\red(L)}
\newrobustcmd\redK{\red(K)}
\newrobustcmd\lcpK{\lcp(K)}
\newrobustcmd\lcpL{\lcp(L)}

\newrobustcmd\suff{\kl[\suff]{\mathsf{suff}}}
\newrobustcmd\inv{{\kl[\inv]{\mathsf{rev}}}}

\newrobustcmd\swap{\mathsf{swap}}

\renewrobustcmd\Im{\kl[\Im]{\mathrm{Im}}}

\newrobustcmd\catObsAutoLSetop{\kl[\catObs]{\mathsf{Obs}(\langL_{\Setop})}}
\newrobustcmd\catReachAutoLSet{\kl[\catReach]{\mathsf{Reach}(\langL_{\Set})}}
\newrobustcmd\catObsAutoLC{\kl[\catObs]{\mathsf{Obs}(\langL)}}
\newrobustcmd\catReachAutoLC{\kl[\catReach]{\mathsf{Reach}(\langL)}}
\newrobustcmd\catReachK{\kl[\catReach]{\mathsf{Reach}(\catK)}}
\newrobustcmd\catObsK{\kl[\catObs]{\mathsf{Obs}(\catK)}}

\title{Automata Minimization: a Functorial Approach\footnote{This work
    was supported by the European Research Council (ERC) under the
    European Union’s Horizon 2020 research and innovation programme
    (grant agreement No.670624), and by the DeLTA ANR project
    (ANR-16-CE40-0007). The authors also thank the Simons Institute
    for the Theory of Computing where this work has been partly
    developed.}  \footnote{This is the "knowledge@knowledged" enriched version
    of the same paper published in the proceedings of CALCO 2017.}}

\author{
  Thomas Colcombet and
  Daniela Petri{\c s}an\\
  {\small CNRS, IRIF, Univ. Paris-Diderot, Paris 7, France}\\
  \texttt{\{thomas.colcombet,petrisan\}@irif.fr}
}

\date{}

\begin{document}

\maketitle

\begin{abstract}
  In this paper we regard languages and their acceptors -- such as
  deterministic or weighted automata, transducers, or monoids -- as
  functors from input categories that specify the type of the
  languages and of the machines to categories that specify the type of
  outputs.

  Our results are as follows:
  a) We provide sufficient conditions on the output category so that
  minimization of the corresponding automata is guaranteed.
  b) We show how to lift adjunctions between the categories for output
  values to adjunctions between categories of automata.
  c) We show how this framework can be applied to several
  phenomena in automata theory, starting with determinization and
  minimization (previously studied from a coalgebraic
  and duality theoretic perspective). We apply in particular  these techniques to
  Choffrut's minimization algorithm for subsequential
  transducers and revisit Brzozowski's minimization algorithm.
\end{abstract}

\section{Introduction}
\label{sec:intro}

There is a long tradition of interpreting results of automata theory
using the lens of category theory.  Typical instances of this scheme
interpret automata as algebras (together with a final map) as put
forward in~\cite{ArbibManes75,goguen1972,AdamekTrnkova89}, or as
coalgebras (together with an initial map), see for
example~\cite{Jacobs97atutorial}. This dual narrative proved very
useful~\cite{BonchiBHPRS14} in explaining at an abstract level
"Brzozowski's minimization@Brzozowski's minimization algorithm"
"algorithm@Brzozowski's minimization algorithm" and the duality between
reachability and observability (which goes back all the way to the
work of Arbib, Manes and Kalman).

In this paper, we adopt a slightly different approach, and we define
directly the notion of an "automaton" (over finite words) as a functor
from a "category representing input words", to a category representing
the computation and output spaces.  The notions of a "language" and of
a "language accepted" by an automaton are adapted along the same
pattern.

We provide several developments around this idea. First, we recall
(see~\cite{icalp2017-submission}) that the existence of a "minimal
automaton" for a "language" is guaranteed by the existence of an
"initial" and a "final automaton" in combination with a "factorization
system". Additionally, we explain how, in the functor presentation
that we have adopted, the existence of "initial" and "final automata"
"for a language" can be phrased in terms of "Kan extensions".  We also
show how adjunctions between categories can be lifted to the level of
"automata for a language" in these categories
(Lemma~\ref{lem:lifting-adjunctions}). This lifting accounts for
several constructions in automata theory, determinization to start
with.  We then use this framework in the explanation of two well-known
constructions in automata theory.

The most involved contribution (Theorem~\ref{theorem:minimization
  transducer}) is to rephrase the minimization result of Choffrut for
"subsequential transducers" in this framework. We do this by
instantiating the category of outputs with the Kleisli category for
the monad $TX=B^*\times X+1$, where $B$ is the output alphabet of the
transducers. In this case, despite the lack of completeness of the
ambient category, one can still prove the existence of an initial and
final automaton, as well as, surprisingly, of a factorization system.

The second concrete application is a proof of correctness of
Brzozowski's minimization algorithm.  Indeed, determinization of
automata can be understood as lifting the Kleisli adjunction between
the categories $\Rel$ (of sets and relations) and $\Set$ (of sets and
functions); and reversing a nondeterministic automaton can be
understood as lifting the self-duality of $\Rel$.  In
Section~\ref{sec:brzozowski} we show how Brzozowski's minimization
algorithm can be explained by combining several liftings of
adjunctions, and in particular as lifting the adjunction between
$\Set$ and its opposite category $\Set^\op$, thus recovering results
from~\cite{BonchiBHPRS14}.

\subparagraph*{Related work.}

Many of the constructions outlined here have already been explained
from a category-theoretic perspective, using various techniques. For
example, the relationship between minimization and duality was subject
to numerous papers, see~\cite{Bezhanishvili2012,BonchiBHPRS14,
  BonchiBRS12} and the references therein.  The coalgebraic
perspective on minimization was also emphasised in papers such
as~\cite{AdamekBHKMS:2012,SilvaEtAl:genPow}.  Understanding
determinization and codeterminization, as well as studying trace
semantics via lifting of adjunctions to coalgebras was considered
in~\cite{JacobsSS15,KerstanKW14}, and is related to our results from
Section~\ref{sec:det-codet-adj}. Subsequential transducers were
understood coalgebraically in~\cite{Hansen10}.

The paper which is closest in spirit to our work is a seemingly
forgotten paper~\cite{Bainbridge74}. However, in this work, Bainbridge
models the \emph{state space} of the machines as a functor.  Left and
right Kan extensions are featured in connection with the initial and
final automata, but in a slightly different
setting. Lemma~\ref{lem:lifting-adjunctions}, which albeit technically
simple, has surprisingly many applications, builds directly on his
work.

\newrobustcmd\alphabetA{A} \newrobustcmd\alphabetB{B}

\section{Languages and Automata as Functors}
\label{sec:lang-auto-functors}

In this section, we introduce the notion of automata via functors,
which is the common denominator of the different contributions of the
paper. We introduce this definition starting from the special case of
classical "deterministic automata@dfaclassic".

\AP In the standard definition, a ""deterministic
automaton@dfaclassic"" is a tuple:
\begin{align*}
  \langle Q,\alphabetA,q_0,F,\delta\rangle
\end{align*}
where $Q$ is a set of ""states@statesclassic"", $\alphabetA$ is an
alphabet (not necessarily finite), $q_0\in Q$ is the ""initial
state@initialclassic"", $F\subseteq Q$ is the set of ""final
states@finalclassic"", and $\delta_a\colon Q\rightarrow Q$ is the
""transition map@transition classic"" for all
letters~$a\in\alphabetA$.  The semantics of an automaton is given by
defining what is a ""run@runclassic"" over an input
word~$u\in\alphabetA^*$, and whether it is
"accepting@acceptingclassic" or not.  Given a word~$e=a_1\dots a_n$,
the "automaton@dfaclassic" ""accepts@acceptsclassic"" the word
if~$\delta_{a_n}\circ\dots\circ\delta_{a_1}(q_0)\in F$, and otherwise
\reintro[acceptsclassic]{rejects it}.

If we see $q_0$ as a map~$\mathit{init}$ from the one element
set~$1=\{0\}$ to~$Q$, that maps~$0$ to~$q_0$, and $F$ as a map
$\mathit{final}$ from~$Q$ to the set~$2=\{0,1\}$, where~$1$ means
`accept' and $0$ means `reject', then the semantics of the
"automaton@dfaclassic" is to associate to each word~$u=a_1\dots a_n$
the map from~$1$ to~$2$ defined
as~$\mathit{final}\circ\delta_{a_n}\circ\cdots\circ\delta_{a_1}\circ\mathit{init}$. If
this map is (constant equal to) $1$, this means that the word is
accepted, and otherwise it is rejected.

\noindent\begin{minipage}{10.5cm}
  ""@\objectLeft""%
  ""@\objectRight""%
  ""@\objectCenter""%
  Pushing this idea further, we can see the semantics of the automaton
  as a functor from the free category generated by the graph on the
  right to $\Set$, and more precisely one that sends the object
  $\objectLeft$ to~$1$ and~$\objectRight$ to~$2$.  In the above
  category, the arrows from $\objectLeft$ to $\objectRight$ are of the
  form $\symbLeft w\symbRight$ for $w$ an arbitrary word in
  $\alphabetA^*$.
\end{minipage} ~~
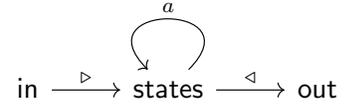
\begin{tikzcd}
  \objectLeft\arrow[r,"\symbLeft"] &
  \objectCenter\arrow[loop,looseness=6,swap,
  "a"]\arrow[r,"\symbRight"] & \objectRight
\end{tikzcd}

\AP Furthermore, since a ""language@languageclassic"" can be seen as a
map from~$\alphabetA^*$ to~$1\rightarrow 2$, we can model it as a
functor from the full subcategory on objects $\objectLeft$ and
$\objectRight$ to the category $\Set$, which maps $\objectLeft$ to~$1$
and~$\objectRight$ to~$2$.

\AP In this section we fix an arbitrary small category $\catI$ and a
full subcategory $\catO$. We denote by $\iota$ the inclusion functor
\[
\begin{tikzcd}
  \catO\arrow[r,hook,"\iota"] & \catI\,.
\end{tikzcd}
\]
We think of $\catI$ as a specification of the inner computations that
an automaton can perform, including black box behaviour, not
observable from the outside.
On the other hand, the full subcategory $\catO$ specifies the
observable behaviour of the automaton, that is, the language it
accepts.
In this interpretation, a machine/automaton $\autA$ is a functor from
$\catI$ to a category of outputs $\catC$, and the ``behaviour'' or
``language'' of $\autA$ is the functor $\langL(\autA)$ obtained by
precomposition with the inclusion $\begin{tikzcd}
  \catO\ar[hook]{r}{\iota}&\catI
\end{tikzcd}$. We obtain the following definition:
\noindent

\begin{definition}[languages and the categories of automata for them]\label{def:aut-accepts-lang}~\\
  \begin{minipage}{13cm}
    A ""$\catC$-language"" is a functor $\langL\colon\catO\to\catC$
    and a ""$\catC$-automaton"" is a functor
    $\autA\colon\catI\to\catC$.  A "$\catC$-automaton" $\autA$
    ""accepts"" a "$\catC$-language" $\langL$ when
    $\autA\circ\iota=\langL$; i.e. the diagram to the right commutes:
  \end{minipage}~~~~
  $
  \begin{tikzcd}
    \catO\ar{r}{\langL}\ar[hook]{d}[swap]{\iota} & \catC \\
    \catI\ar{ru}[swap]{\autA} &
  \end{tikzcd}
  $\\
  \noindent
  We write $\intro*\catAutoL$ for the subcategory of the functor
  category $[\catI,\catC]$ where
  \begin{enumerate}
  \item objects are "$\catC$-automata" that "accept" $\langL$.
  \item arrows are natural transformations $\alpha\colon\autA\to\autB$
    so that the natural transformation obtained by composition with
    the inclusion functor $\iota$ is the identity natural
    transfomation on $\langL$, that is, $\alpha\circ
    \iota=\mathit{id}_\langL$.
  \end{enumerate}
\end{definition}

\subsection{Minimization of "$\catC$-automata"}
\label{sec:minim-catAuto}

In this section we show that the notion of a "minimal automaton" is an
instance of a more generic notion of minimal object that can be
defined in an arbitrary category $\catK$ whenever there exist an
"initial object", a "final object", and a "factorization system"
$(\Epi,\Mono)$.

\AP Let $X,Y$ be two objects of $\catK$.  We say that:
\begin{center}
  $X$\quad$(\Epi,\Mono)$""-divides""\quad$Y$\qquad if\qquad $X$ is an
  \Epi""-quotient"" of an \Mono""-subobject"" of~$Y$.
\end{center}

Let us note immediately that in general this notion of
(\Epi,\Mono)"-divisibility" may not be transitive%
\footnote{There are nevertheless many situations for which it is the
  case; In particular when the category is "regular", and $\Epi$
  happens to be the class of "regular epis". This covers in particular
  the case of all "algebraic categories" with \Epi"-quotients" being
  the standard "quotients of algebras", and \Mono"-subobjects" being
  the standard "subalgebras".}.  \AP It is now natural to define an
"object"~$M$ to be (\Epi,\Mono)""-minimal"" in the "category", if it
(\Epi,\Mono)"-divides" all "objects" of the "category". Note that
there is no reason a priori that an (\Epi,\Mono)"-minimal" object in a
"category", if it exists, be unique up to "isomorphism". Nevertheless,
in our case, when the category has both "initial@initial object" and a
"final object", we can state the following minimization lemma:

\begin{lemma}\AP\label{lemma:minimal}%
  Let~$\catK$ be a "category" with "initial object"~$I$ and "final
  object"~$F$ and let $(\Epi,\Mono)$ be a "factorization system"
  for~$\catK$. Define for every "object"~$X$:
  \begin{itemize}
  \item $\intro*\Min{}$ to be the "factorization" of the "only
    arrow@initial object" from~$I$ to~$F$,
  \item $\intro*\Reach(X)$ to be the "factorization" of the "only
    arrow@initial object" from~$I$ to~$X$, and $\intro*\Obs(X)$ to be
    the "factorization" of the "only arrow@final object" from~$X$
    to~$F$.
  \end{itemize}
  Then
  \begin{itemize}
  \item $\Min{}$ is (\Epi,\Mono)"-minimal", and
  \item $\Min{}$ is "isomorphic" to both $\Obs(\Reach(X))$ and
    $\Reach(\Obs(X))$ for all "objects"~$X$.
  \end{itemize}
\end{lemma}
\begin{proof} The proof essentially consists of a diagram:
  
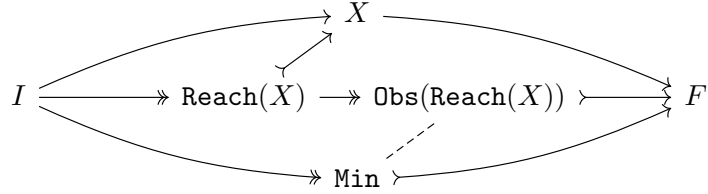
\begin{center}
    \begin{tikzcd}[column sep={1.5cm,between origins},row
      sep={1.1cm,between origins}]
      &&& X\ar[rrrd,bend left=10] & &
      \\
      I\ar[rrru,bend left=10]\ar[rrrd,bend right=10,two
      heads]\ar[rr,two heads] &&
      \Reach(X)\ar[ru,rightarrowtail]\ar[rr,two heads] &&
      \Obs(\Reach(X))\ar[rr, rightarrowtail] && F
      \\
      &&& \Min{} \ar[rrru,bend right=10,rightarrowtail]\ar[ru,dashed,
      no head]& &
    \end{tikzcd}
  \end{center}
  Using the definition of~$\Reach{}$ and~$\Obs{}$, and the fact that
  $\Epi$ is closed under "composition", we obtain that
  $\Obs(\Reach(X))$ is an (\Epi,\Mono)"-factorization" of the "only
  arrow@initial object" from~$I$ to~$F$. Thus, thanks to the "diagonal
  property" of a "factorization system", $\Min{}$
  and~$\Obs(\Reach(X))$ are "isomorphic". Hence, furthermore,
  since~$\Obs(\Reach(X))$ (\Epi,\Mono)"-divides"~$X$ by construction,
  the same holds for~$\Min{}$.  In a symmetric way, $\Reach(\Obs(X))$
  is also "isomorphic" to~$\Min{}$.
\end{proof}

\AP An object $X$ of $\catK$ is called ""reachable"" when $X$ is
isomorphic to $\Reach(X)$. We denote by $\intro*\catReachK$ the full
subcategory of $\catK$ consisting of "reachable" objects. \AP
Similarly, an object $X$ of $\catK$ is called ""observable"" when $X$
is isomorphic to $\Obs(X)$. We denote by $\intro*\catObsK$ the full
subcategory of $\catK$ consisting of "observable" objects.

\AP We can express reachability $\Reach$ and observability $\Obs$ as
the right, respectively the left adjoint to the inclusion of
$\catReachK$, respectively of $\catObsK$ into $\catK$. It is indeed a
standard fact that factorization systems give rise to reflective
subcategories, see~\cite{cassidy_hebert_kelly_1985}. In our case, this
is the reflective subcategory $\catObsK$ of $\catK$. By a dual
argument, the category $\catReachK$ is coreflective in $\catK$. We can
summarize these facts in the next lemma.
\begin{lemma}
  \label{lem:adj-reach-obs}
  Let~$\catK$ be a "category" with "initial object"~$I$ and "final
  object"~$F$ and let $(\Epi,\Mono)$ be a "factorization system"
  for~$\catK$.  We have the adjunctions
\end{lemma}
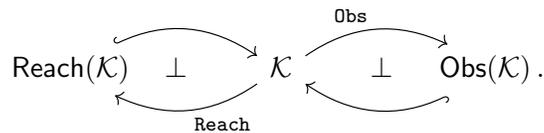
\[
\begin{tikzcd}[column sep={14mm,between origins}]
  \catReachK \ar[rr,bend left,hook] &\bot &\catK\ar[rr,bend left,
  "\Obs"]\ar[ll,bend left, "\Reach"] &\bot~ & \catObsK\,.\ar[ll,bend
  left, hook]
\end{tikzcd}
\]

\AP
In what follows we will instantiate $\catK$ with the category
$\catAutoL$ of $\catC$-automata accepting a language
$\langL$. Assuming the existence of an initial and final automaton for
$\langL$ -- denoted by $\intro*\autAinitL$, respectively $\intro*\autAfinalL$ --
and, of a factorization system, we obtain the functorial version of
the usual notions of "reachable sub-automaton" $\ReachA$ and
"observable quotient automaton" $\ObsA$ of an automaton $\autA$. The
"minimal automaton" $\MinL$ for the language $\langL$ is obtained via
the factorization
\[
\begin{tikzcd}
  \autAinitL \arrow[r,two heads] & \MinL\ar[r,tail]& \autAfinalL\,.
\end{tikzcd}
\]

Lemma~\ref{lemma:minimal} implies that the "minimal automaton" divides
any other automaton recognising the language, while
Lemma~\ref{lem:adj-reach-obs} instantiates to the results
of~\cite[Section~9.4]{BonchiBHPRS14}.

\subsection{Minimization of "$\catC$-automata": sufficient conditions
  on $\catC$}
\label{sec:minimization:suff-cond}
We now can list sufficient conditions on $\catC$ so that the category
of "$\catC$-automata" $\catAutoL$ accepting a "$\catC$-language"
$\langL$ satisfies the three conditions of Lemma~\ref{lemma:minimal}.

We start with the factorization system. It is well known that given a
"factorization system" $(\Epi,\Mono)$ on $\catC$, we can extend it to
a "factorization system"
$(\Epi_{[\catI,\catC]},\Mono_{[\catI,\catC]})$ on the functor category
$[\catI,\catC]$ in a pointwise fashion. That is a natural
transformation is in~$\Epi_{[\catI,\catC]}$ if all its components are
in~$\Epi$, and analogously, a natural transformation is
in~$\Mono_{[\catI,\catC]}$ if all its components are in~$\Mono$. In
turn, the factorization system on the functor category $[\catI,\catC]$
induces a factorization system on each subcategory $\catAutoL$.
\begin{lemma}
  \label{lem:lifting-fact-syst}
  If~$\catC$ has a factorization system $(\Epi,\Mono)$, then the
  category~$\catAutoL$ has a factorization
  system~$(\intro*\EpiAut,\intro*\MonoAut)$, where~$\EpiAut$ consists
  of all the natural transformations with components in~$\Epi$
  and~$\MonoAut$ consists of all natural transformations with
  components in~$\Mono$.
\end{lemma}
The proof of Lemma~\ref{lem:lifting-fact-syst} is the same as the
classical one that shows that "factorization systems" can be lifted to
functor categories.

\begin{lemma}
  If the left Kan extension $\Lan{\langL}{\iota}$ of~$\langL$
  along~$\iota$ exists, then it is an initial object in~$\catAutoL$,
  that is, $\autAinitL$ exists and is isomorphic
  to~$\Lan{\langL}{\iota}$.

  Dually, if the right Kan extension~$\Ran{\langL}{\iota}$ of~$\langL$
  along~$\iota$ exists, then so does the final object~$\autAfinalL$
  of~$\catAutoL$ and $\autAfinalL$ is isomorphic
  to~$\Ran{\langL}{\iota}$.
\end{lemma}

\begin{proof}[Proof Sketch]
  Assume the left Kan extension exists. Then the canonical natural
  transformation $\langL\to \Lan{\langL}{\iota}\circ \iota$ is an
  isomorphism since $\iota$ is full and faithful. Whenever $\autA$
  accepts $\langL$, that is, $\autA\circ\iota=\langL$, we obtain the
  required unique morphism $\Lan{\langL}{\iota}\to\autA$ using the
  universal property of the Kan extension. The argument for the right
  Kan extension follows by duality.
\end{proof}

\begin{corollary}
  \label{cor:exist-Kan-suff-cond}
  Assume $\catC$ is complete, cocomplete and has a factorization
  system and let $\langL$ be a "$\catC$-language". Then the "initial
  $\langL$-automaton" and the "final $\langL$-automaton" exist and are
  given by the left, respectively right Kan extensions of $\langL$
  along $\iota$. Furthermore, the minimal $\catC$-automaton $\Min(\langL)$
  accepting $\langL$ is obtained via the factorization $
  \begin{tikzcd}
    \Lan{\langL}{\iota}\ar[r,two heads] & \Min(\langL)\ar[r,tail] &
    \Ran{\langL}{\iota}\,.
  \end{tikzcd}
  $
\end{corollary}

\begin{remark}
  Depending on the category $\catI$, we may relax the conditions in
  Corollary~\ref{cor:exist-Kan-suff-cond},
  see~Lemma~\ref{lem:the-minimization-wheel}.  Furthermore, we
  emphasise that these conditions are only sufficient. In
  Section~\ref{sec:choffrut} we consider the example of subsequential
  transducers and we instantiate $\catC$ with a Kleisli
  category. Although this category does not have powers, the final
  automaton exists.
\end{remark}

\section{Word Automata}
\label{sec:word-input-category}

\noindent
\phantomintro\objectLeft\phantomintro\objectRight\phantomintro\objectCenter
\begin{minipage}{10.5cm}
  Hereafter, we restrict our attention to the case of word automata,
  for which the input category $\catI$ is the three-object category
  with arrows spanned by $\intro*\symbLeft$, $\intro*\symbRight$ and
  $a$ for all $a\in\alphabetA$, as in the diagram on the right and
  where the
\end{minipage}~~~
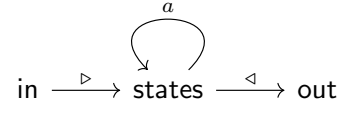
\begin{tikzcd}
  \reintro*\objectLeft\arrow[r,"\reintro*\symbLeft"] &
  \reintro*\objectCenter\arrow[loop,looseness=6,swap,
  "a"]\arrow[r,"\reintro*\symbRight"] & \reintro*\objectRight
\end{tikzcd}

\noindent composite of $
\begin{tikzcd}
  \objectCenter\ar[r,"w"] & \objectCenter\ar[r,"w'"] & \objectCenter
\end{tikzcd}
$ is given by the concatenation $ww'$.

\AP Let $\catO$ be the full subcategory of $\catI$ on objects
$\objectLeft$ and $\objectRight$.  We consider "$\catC$-languages",
which are now functors $\langL\colon\catO\to\catC$.  If
$\langL(\objectLeft)=X$ and $\langL(\objectRight)=Y$ we call $\langL$
a ""$(\catC, X,Y)$-language"".  Similarily, we consider
"$\catC$-automata" that are functors $\autA\colon\catI\to\catC$.  If
$\autA(\objectLeft)=X$ and $\autA(\objectRight)=Y$ we call $\autA$ a
""$(\catC, X,Y)$-automaton"".

The next lemma refines Corollary~\ref{cor:exist-Kan-suff-cond}.
It appeared in~\cite{Ballester-Bolinches15} in the "$(\Set,1,2)$-language" case.
\begin{lemma}[from~\cite{icalp2017-submission}]
  \label{lem:the-minimization-wheel}
  If $\catC$ has countable products and countable coproducts, and a
  factorization system, then the minimal "$\catC$-automaton"
  "accepting" $\langL$ is obtained via the "factorization" in the next
  diagram.
\end{lemma}
\noindent
\begin{minipage}[c]{10cm}
  The initial automaton has as state space the copower
  $\coprod\limits_{u\in
    A^*}\langL(\objectLeft)$. In~\cite{icalp2017-submission} we gave a
  direct proof of initiality, but here we can also notice that this is
  exactly what the colimit computation of the left Kan extension of
  $\langL$ along $\iota$ yields -- using the fact that there are no
  morphisms from $\objectRight$ to $\objectCenter$ in $\catI$ and the
  only morphism on which you take the colimit are of the form
  $\symbLeft w\colon \objectLeft\to\objectCenter$ for all
  $w\in\alphabetA^*$.
\end{minipage}
~~~
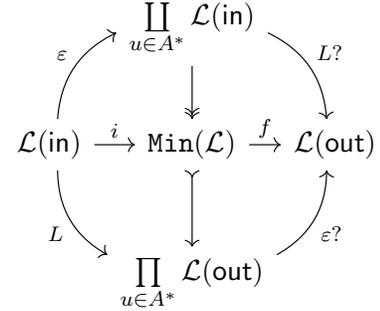
\begin{minipage}[b]{5cm}
  \begin{tikzcd}[column sep={1.9cm,between origins},row
    sep={1.7cm,between origins}]
    & \coprod\limits_{u\in A^*}\langL(\objectLeft)
    \arrow[rd,bend left, "L ?"]\arrow[d,two heads] &
    \\
    \langL(\objectLeft)
    \arrow[bend right,swap,"L"]{rd} \arrow[ru,,bend left,
    "\varepsilon"] \arrow[r,"i"] & \Min(\langL) \arrow[r,"f"] \arrow[d,
    tail] & \langL(\objectRight)
    \\
    & \prod\limits_{u\in A^*}\langL(\objectRight) \arrow[bend
    right]{ru}[swap]{\varepsilon ?} &
  \end{tikzcd}%
\end{minipage}

\subsection{Lifting Adjunctions to Categories of Automata}
\label{sec:lifting-adjucntions-to-automata}

In this section we will juggle with "languages" and "automata"
interpreted over different categories connected via "adjunctions".

Assume we have an adjunction between two categories $\catC$ and
$\catD$
\knowledgeconfigure{quotation=false}%
\[
\begin{tikzcd}
  \catC \ar[rr,shift left=2, "F"] & \bot &\catD\ar[ll,shift left=2,
  "G"]\,,
\end{tikzcd}
\]
\knowledgeconfigure{quotation=true}%
with $F\dashv G\colon \catD\to\catC$.  Let $(-)^*$ and $(-)_*$ denote
the induced natural isomorphisms between the homsets.  In particular,
given objects $I$ in $\catC$ and $O$ in $\catD$, we have bijections
\knowledgeconfigure{quotation=false}%
\begin{equation}
  \label{eq:bij-lang}
  \begin{tikzcd}[column sep=large]
    \catC(I,GO)\ar[rr,shift left=2, "(-)^*"] &
    &\catD(FI,O)\ar[ll,shift left=2, "(-)_*"]
  \end{tikzcd}
\end{equation}
\knowledgeconfigure{quotation=true}%
These bijections induce a one-to-one correspondence between
"$(\catC,I,GO)$-languages" and
"$(\catD,FI,O)$-@$(\catD,FI,O)$-languages"
"languages@$(\catD,FI,O)$-languages", which by an abuse of notation we
denote by the same symbols:
\[
\begin{tikzcd}[column sep=large]
  \text{\kl{$(\catC,I,GO)$-languages}}\ar[rr, shift left=2, "(-)^*"] &
  &\text{\kl{$(\catD,FI,O)$-languages}}\ar[ll,shift left=2, "(-)_*"]
\end{tikzcd}
\]

Indeed, given a "$(\catC,I,GO)$-language" $\langL\colon \catO\to\catC$
we obtain a "$(\catD,FI,O)$-language" $\langL^*\colon \catO\to\catD$
by setting $ \langL^*(\symbLeft w\symbRight) = (\langL(\symbLeft
w\symbRight))^* \in \catD(FI,O)$. Conversely, given a
"$(\catD,FI,O)$-language" $\langL'$ we obtain a
"$(\catC,I,GO)$-language" $(\langL')_*$ by setting
$(\langL')_*(\symbLeft w\symbRight)=(\langL'(\symbLeft
w\symbRight))_*$.

\begin{lemma}
  \label{lem:lifting-adjunctions}
  Assume $\langL_\catC$ and $\langL_\catD$ are "$(\catC,I,GO)$-@$(\catC,I,GO)$-languages",
  respectively "$(\catD,FI,O)$-languages" so that
  $\langL_\catD=(\langL_\catC)^*$. Then the adjunction $F\dashv G$
  lifts to an adjunction
  $\overline{F}\dashv\overline{G}\colon\catAutoLD\to\catAutoLC$. The
  lifted functors $\overline{F}$ and $\overline{G}$ are defined as
  $F$, resp. $G$ on the state object, that is, the following diagram
  commutes
  \knowledgeconfigure{quotation=false}%
  \[
  \begin{tikzcd}[column sep=large,row sep=5mm]
    \catAutoLC\ar[dd,"\St"] \ar[rr,bend left=12, "\overline{F}"] & \bot &\catAutoLD\ar[ll,bend left=12, "\overline{G}"]\ar[dd,"\St"]\\
    & & \\
    \catC \ar[rr,bend left=14, "F"] & \bot &\catD\ar[ll,bend left=14,
    "G"]
  \end{tikzcd}
  \]
  \knowledgeconfigure{quotation=true}%
  where the functor $\St\colon\catAutoLC\to\catC$ is the evaluation at
  $\objectCenter$, that is, it sends an automaton
  $\autA\colon\catI\to\catC$ to $\autA(\objectCenter)$.
\end{lemma}
\begin{proof}[Proof sketch]
  The functor $\overline{F}$ maps an automaton
  $\autA\colon\catI\to\catC$ from $\catAutoLC$ to the
  $\catD$-automaton $\overline{F}\autA\colon \catI\to\catD$ mapping
  $\symbLeft\colon\objectLeft\to\objectCenter$ to $F\autA(\symbLeft)$,
  $a\colon\objectCenter\to\objectCenter$ to $F(\autA(a))$ and
  $\symbRight\colon\objectCenter\to\objectRight$ to the adjoint
  transpose $(\autA(\symbRight))^*$ of $\autA(\symbRight)$. The
  functor $\overline{G}$ is defined similarly.
\end{proof}

\section{Choffrut's minimization of subsequential transducers}
\label{sec:choffrut}

In~\cite{Choffrut79,Choffrut2003} Choffrut establishes a minimality
result for "subsequential transducers", which are deterministic
automata that output a word while processing their input.  In this
section, we show the existence of minimal subsequential transducers  using our functorial
framework.

We first present the model of "subsequential transducers" in
Section~\ref{subsection:subsequential}, show how these can be
identified with "automata in the Kleisli category of a suitably chosen
monad", and state the minimization result,
Theorem~\ref{theorem:minimization transducer}. The subsequent sections
provide the necessary material for proving the theorem.

\subsection{Subsequential transducers and automata in a Kleisli
  category}
\label{subsection:subsequential}

"Subsequential transducers" are (finite state) machines that compute
partial functions from input words in some alphabet~$\alphabetA$ to
output words in some other alphabet~$\alphabetB$.  In this section, we
recall the classical definition of these objects, and show how it can
be phrased categorically.

\begin{definition}
  A ""subsequential transducer"" is a tuple
  \begin{align*}
    T=(Q, \alphabetA, \alphabetB, q_0, t, u_0,(-\cdot
    a)_{a\in\alphabetA},(-*a)_{a\in\alphabetA})\ ,
  \end{align*}
  where
  \begin{itemize}
  \item $\alphabetA$ is the ""input alphabet@input alphabetSS"" and
    $\alphabetB$ the ""output one@output alphabetSS"",
  \item $Q$ is a (finite) set of ""states@statesSS"".
  \item $q_0$ is either undefined or belongs to~$Q$ and is called the
    ""initial state of the transducer"".
  \item $t\colon Q\rightharpoonup \alphabetB^*$ is a ""partial
    termination function"".
  \item $u_0\in \alphabetB^*$ is defined if and only if~$q_0$ is, and
    is the ""initialization value"".
  \item $-\cdot a\colon Q\rightharpoonup Q$ is the ""partial
    transition function for the letter~$a$"", for all $a\in A$.
  \item $-* a\colon Q\rightharpoonup \alphabetB^*$ is the ""partial
    production function for the letter~$a$"" for all $a\in
    \alphabetA$; it is required that $q* a$ be defined if and only if
    $(q\cdot a)$ is.
  \end{itemize}
  A "subsequential transducer" $T$ computes a partial function
  $\intro*\semTrans T\colon \alphabetA^*\rightharpoonup \alphabetB^*$
  defined as:
  \begin{align*}
    \semTrans T(a_1\dots a_n)&=u_0(q_0*a_1)(q_1*a_2)\dots
    (q_{n-1}*a_n)t(q_n)&\text{for all~$a_1\dots a_n\in \alphabetA^*$,}
  \end{align*}
  where for each $1\le i\le n$ either $q_i$ is undefined or belongs
  to~$Q$ and is given by $q_i=q_{i-1}\cdot a_i$.  Furthermore,
  $\semTrans T(a_1\dots a_n)$ is undefined when at least one of
  $q_0,\ldots,q_n$ or $t(q_n)$ is so.
\end{definition}

\AP These "subsequential transducers" are modeled in our framework as
"automata in the category" of free algebras for the
monad~$\monadTrans$, that we describe now.
\begin{definition}
  The monad $\intro*\monadTrans\colon\Set\to\Set$ is defined by
  \[
  \monadTrans(X)=B^*\times X + 1
  \]
  with unit $\eta_X$ and multiplication $\mu_X$ defined for all $x\in
  X$ and $w,u\in B^*$ as:
  \begin{align*}
    &&
    \mu_X\colon\qquad\monadTrans^2 X&\to \monadTrans X\\
    \eta_X\colon\quad X &\to B^*\times X +1&
    (w,(u,x))&\mapsto (wu, x)\\
    x&\mapsto (\varepsilon, x)&
    (w,\bot) &\mapsto \bot\\
    && \bot &\mapsto \bot
  \end{align*}
  where we denote by $\bot$ the unique element of $1$ (used to model
  the partiality of functions).
\end{definition}

\AP Recall that the category of free $\monadTrans$-algebras is the
""Kleisli category for $\monadTrans$"", $\intro*\KlTrans$, that has as
objects sets $X,Y, \ldots$ and as ""morphisms@Kleisli morphism""
$f\colon X\tokl Y$ functions $f\colon X\to B^*\times Y+1$ in $\Set$,
that is a partial function from $X$ to $B^*\times Y$.

\AP Let $T$ be a "subsequential transducer". The "initial state of the
transducer" $q_0$ and the "initialization@initialization value"
"value@initialization value" $u_0$ together form a morphism $i\colon
1\tokl Q$ in the category $\KlTrans$. Similarly, the "partial
transition@partial transition function" "function@partial transition
function" and the "partial production function" for a letter $a$ of
the input alphabet $A$ are naturally identified to "Kleisli morphisms"
$\delta_a\colon Q\tokl Q$ in $\KlTrans$. Finally, the "partial
termination function" together with the "partial production function"
are nothing but a "Kleisli morphism" of the form $t\colon Q\tokl
1$. To summarise, we obtained that a "subsequential transducer" $T$ in
the sense of~\cite{Choffrut2003} is specified by the following
"morphisms@Kleisli morphism" in $\KlTrans$:
\[
\begin{tikzcd}
  1\arrow[r,negated,"i"] & Q\arrow[loop,looseness=6,negated,swap,
  "\delta_a"]\arrow[r,negated,"t"] & 1
\end{tikzcd}
\]
that is by a functor $\autA_T\colon \catI\to\KlTrans$ or equivalently,
a "$(\KlTrans,1,1)$-automaton". The subsequential function realised by
the transducer $T$ is a partial function $A^*\rightharpoonup B^*$ and
is fully captured by the "$(\KlTrans,1,1)$-language" $\langL_T\colon
\catO\to\KlTrans$ accepted by $\autA_T$, which is obtained as
$\autA_T\circ\iota$. Indeed, this $\KlTrans$-language gives for each
word $w\in A^*$ a "Kleisli morphism" $\langL_T(\symbLeft
w\symbRight)\colon 1\tokl 1$, or equivalently, outputs for each word
in $A^*$ either a word in $B^*$ or the undefined element $\bot$.
 
Putting all this together, we can state the following lemma, which
validates the categorical encoding of "subsequential transducers":
\begin{lemma}%
  \intro[{$(\KlT , 1,1)$-automaton}]{}%
  \intro[$\KlTrans $-automaton]{}%
  "Subsequential transducers" are in one to one correspondence with
  "$(\KlTrans,1,1)$-automata", and partial maps from~$A^*$ to~$B^*$
  are in one to one correspondence with "$(\KlTrans,1,1)$-languages".
  Furthermore, the acceptance of languages is preserved under these
  bijections.
\end{lemma}

In the rest of this section we will see how to obtain Choffrut's
minimization result as an application of Lemma~\ref{lemma:minimal}.
I.e., we have to provide in the category of
"$(\KlTrans,1,1)$-automata",
\begin{enumerate}
\item an "initial object@initial transducer",
\item a "final object@final transducer", and,
\item a "factorization system@factorization transducer".
\end{enumerate}

The existence of the "initial transducer" is addressed in
Section~\ref{subsection:initial transducer}, the one of the "final
transducer" is the subject of Section~\ref{subsection:final
  transducer}.  In Section~\ref{subsection:factorization transducers}
we show how to construct a "factorization system@factorization
transducer".  Together, we obtain:
\begin{theorem}[Categorical version of~\cite{Choffrut79,Choffrut2003}]
  \label{theorem:minimization transducer}
  For every~"$(\KlTrans,1,1)$-language", there exists a minimal
  "$(\KlTrans,1,1)$-automaton" for it.
\end{theorem}
Let us note that only the existence of the automaton is mentioned in
this statement, and the way to compute it effectively is not addressed
as opposed to Choffrut's work.  Nevertheless,
Lemma~\ref{lemma:minimal} describes what are the basic functions that
have to be implemented, namely $\Reach$ and $\Obs$.

\AP The rest of this section is devoted to establishing the three
above mentioned points.  Unfortunately, as it is usually the case with
Kleisli categories, $\KlTrans$ is neither complete, nor cocomplete. It
does not even have binary products, let alone countable powers.  Also,
the existence of a "factorization system" does not generally hold in
Kleisli categories. Hence, providing the above three pieces of
information requires a bit of work.

In the next section we present an adjunction between categories of
"$(\KlTrans,1,1)$-automata" and "$(\Set,1,\alphabetB^*)$-automata"
which is then used in the subsequent ones for proving the existence of
"initial@initial transducer" and "final automata@final transducer". We
finish the proof with a presentation of the "factorization
systems@factorization transducer".

\subsection{Back and forth to automata in $\Set$}
\label{subsection:transducer adjunction}

In order to understand what are the properties of the category of
"$(\KlTrans,1,1)$-automata", an important tool will be the ability to
see alternatively a "subsequential transducer" as an "automaton in
$\KlTrans$" as described above, or as an "automaton in $\Set$",
since $\Set$ is much better behaved than $\KlTrans$.  These two points
of view are related through an adjunction, making use of the results
of Section~\ref{sec:lifting-adjucntions-to-automata} and
Lemma~\ref{lem:lifting-fact-syst}.

\AP Indeed, we start from the well known adjunction between $\Set$ and
$\KlTrans$:
\begin{equation}
  \label{eq:adj-kleisli}
  \begin{tikzcd}
    \Set \ar[rr,bend left, "\FreeTrans"] & \bot &\KlTrans\ar[ll,bend
    left, "\UTrans"]\,.
  \end{tikzcd}
\end{equation}
We recall that the free functor $\intro*\FreeTrans$ is defined as the
identity on objects, while for any function $f\colon X\to Y$ the
"morphism@Kleisli morphism" $\FreeTrans f\colon X\tokl Y$ is defined
as $\eta_Y\circ f\colon X\to \monadTrans Y$. For the other direction,
the functor $\intro*\UTrans$ maps an object $X$ in $\KlTrans$ to
$\monadTrans X$ and a morphism $f\colon X\tokl Y$ (which is seen here
as a function $f\colon X\to \monadTrans Y$) to $\mu_Y\circ \monadTrans
f\colon \monadTrans X\to \monadTrans Y$.

\AP\intro[$(\Set , 1,B^*+1)$-automaton]{}%
A simple, yet important observation is that the language of interest,
which is a partial function $L\colon A^*\rightharpoonup B^*$ can be
modeled either as a "$(\KlT, 1,1)$-language" $\langLKlT$, or
equivalently, as a "$(\Set, 1,B^*+1)$-language" $\langLSet$. This is
because for each $w\in A^*$ we can identify $L(w)$ either with an
element of $\KlTrans(1,1)$ or, equivalently, as an element of
$\Set(1,B^*+1)$.
\begin{align*}
  \langLKlT\colon\qquad\catO&\to\KlTrans &     \langLSet\colon\qquad\catO&\to\Set \\
  \objectLeft&\mapsto 1 &     \objectLeft&\mapsto 1 \\
  \objectRight&\mapsto 1 &     \objectRight&\mapsto B^*+1 \\
  \symbLeft w \symbRight &\mapsto L(w)\colon 1\tokl 1 & \symbLeft w
  \symbRight &\mapsto L(w)\colon 1\to B^*+1
\end{align*}
To see how this fits in the scope of
Section~\ref{sec:lifting-adjucntions-to-automata}, notice that
$\langLKlT$ is a "$(\KlT, \FreeTrans 1,1)$-language", while
$\langLSet$ is a "$(\Set, 1,\UTrans 1)$-language" and they correspond
to each other via the
bijections described in~\eqref{eq:bij-lang}.\\
\AP
\begin{minipage}{9.8cm}
  Applying Lemma~\ref{lem:lifting-adjunctions} for the Kleisli
  adjunction~\eqref{eq:adj-kleisli} we obtain an adjunction
  $\FreeAut\dashv \AutUTrans$ between the categories of
  "$\KlTrans$-automata for $\langLKlT$" and of "$\Set$-automata
  accepting $\langLSet$", as depicted in the picture on the right.
  The functor $\AutUTrans$ sends a "$(\KlT, 1,1)$-automaton" with state
  object $Q$ to a "$(\Set, 1,B^*+1)$-automaton" with state object $\monadTrans
  Q$, while $\FreeAut$ sends a "$(\Set, 1,B^*+1)$-automaton" with state object
  $Q'$ to a "$(\KlT, 1,1)$-automaton" with same state object $Q'$.

  We will make heavy use of this correspondence in what follows.
\end{minipage}
~~~$
\begin{tikzcd}[row sep={0.6cm,between origins},column sep=small]
  \catAutoLSet\ar[ddd,"\St"] \ar[rr,bend left=15, "\FreeAut"] & \bot &\catAutoLKlT\ar[ll,bend left=15, "\AutUTrans"]\ar[ddd,"\St"]\\
  & & \\
  & & \\
  \Set \ar[rr,bend left=17, "\FreeTrans"] & \bot &\KlT\ar[ll,bend
  left=17, "\UTrans"]\,.
\end{tikzcd}
$

\subsection{The initial $\KlTrans$-automaton for the language
  $\langLKlT$}
\label{subsection:initial transducer}

\intro[initialSS]{}\intro[initial transducer]{}%
The functor $\FreeAut$ is a left adjoint and consequently preserves
colimits and in particular the "initial object". We thus obtain that
the initial $\langL_{\KlT}$-automaton is
$\FreeAut(\AutInit(\langLSet))$, where $\AutInit(\langLSet)$ is the
"initial object" of $\catAutoLSet$.  This automaton can be obtained by
Lemma~\ref{lem:the-minimization-wheel} as the functor
$\AutInit(\langLSet)\colon\catI\to\Set$ specified by
$\AutInit(\langLSet)(\objectCenter)=A^*$ and for all~$a\in\alphabetA$
\begin{align*}
  \AutInit(\langLSet)(\symbLeft)\colon 1&\to A^* &
  \AutInit(\langLSet)(\symbRight)\colon A^*&\to B^*+1 &
  \AutInit(\langLSet)(a)\colon  A^*&\to A^* \\
  0&\mapsto \varepsilon& w&\mapsto L(w)& w&\mapsto wa
\end{align*}
\AP Hence, by computing the image of $\AutInit(\langLSet)$ under
$\FreeAut$, we obtain the following description of the "initial
$\KlTrans$-automaton@initial automaton" $\AutInit(\langLKlT)$
"accepting" $\langLKlT$: $\AutInit(\langLKlT)(\objectCenter)=A^*$ and
for all $a\in\alphabetA$
\begin{align*}
  \hspace{-5pt}
  \AutInit(\langLKlT)(\symbLeft)\colon 1&\tokl A^* &
  \hspace{-5pt}
  \AutInit(\langLKlT)(\symbRight)\colon A^*&\tokl 1&
  \hspace{-5pt}
  \AutInit(\langLKlT)(a)\colon  A^*&\tokl A^*\\
  0&\mapsto (\varepsilon,\varepsilon)& w&\mapsto L(w)& w&\mapsto
  (\varepsilon,wa)
\end{align*}

\subsection{The final $\KlTrans$-automaton for the language
  $\langLKlT$}
\label{subsection:final transducer}%

\intro[final transducer]{}%
The case of the final $\KlTrans$-automaton is more complicated, since
it is not constructed as easily. However, assuming the final automaton
exists, it has to be sent by $\AutUTrans$ to a final
"$\Set$-automaton". Moreover, by
Lemma~\ref{lem:UAut-reflects-final-obj}, in order to prove that a
given $\KlTrans$-automaton $\autA$ is a final object of $\catAutoLKlT$
it suffices to show that $\AutUTrans(\autA)$ is the final object in
$\catAutoLSet$. The proof of the following lemma generalises the fact
that $\UTrans$ reflects final objects and can be proved in the same
spirit.
\begin{lemma}
  \label{lem:UAut-reflects-final-obj}
  The functor $\AutUTrans\colon \catAutoLKlT\to\catAutoLSet$ reflects
  final objects.
\end{lemma}

The final object in $\catAutoLSet$ is the automaton
$\AutFin(\langLSet)$ as described using
Lemma~\ref{lem:the-minimization-wheel}.  The functor
$\AutFin(\langLSet)\colon\catI\to\Set$ specified by
\begin{align*}
  &\AutFin(\langLSet)(\objectCenter)=(B^*+1)^{A^*}&
  &\begin{array}{rl}
    \AutFin(\langLSet)(\symbRight)\colon  (B^*+1)^{A^*}&\to B^*+1\\
    K&\mapsto K(\varepsilon)
  \end{array}\\
  &\begin{array}[c]{rl}
    \AutFin(\langLSet)(\symbLeft)\colon 1&\to (B^*+1)^{A^*}\\
    0&\mapsto L
  \end{array}&
  &\begin{array}{rl}
    \AutFin(\langLSet)(a)\colon (B^*+1)^{A^*}&\to (B^*+1)^{A^*}\\
    K&\mapsto \lambda w. K(aw)
  \end{array}
\end{align*}

To describe the set of states of the final automaton in $\catAutoLKlT$
we need to introduce a few notations. Essentially we are looking for a
set of states $Q$ so that $B^*\times Q+1$ is isomorphic to
$(B^*+1)^{A^*}$. The intuitive idea is to decompose each function in
$K\in (B^*+1)^{A^*}$ (\AP except for the one which is nowhere defined,
that is the function $\intro*\botfun=\lambda w.\bot$) into a word in
$B^*$, the common prefix of all the $B^*$-words in the image of $K$,
and an "irreducible function".

\AP For~$v\in B^*$ and a function $K\neq \botfun$ in $(B^*+1)^{A^*}$,
denote by $v\intro*\pstar K$ the function defined for
all~$u\in\alphabetA^*$ by $ (v\pstar K)(u)=v\,K(u)$ if $K(u)\in B^*$
and $(v\pstar K)(u)=\bot$ otherwise.
\AP Define also the ""longest common prefix"" of $K$, $\intro*\lcpK\in
B^*$, as the longest word that is prefix of all~$K(u)\neq\bot$ for~$u$
in~$A^*$ (this is well defined since~$K\neq\botfun$).  The ""reduction
of~$K$"", $\intro*\red(K)$, is defined as:
\begin{align*}
  \red(K)(u)&= \begin{cases}
    v&\text{if}~K(u)=\lcp(K)\,v,\\
    \bot&\text{otherwise}.
  \end{cases}
\end{align*}
\AP Finally, $K$ is called ""irreducible"" if $\lcp(K)=\varepsilon$
(or equivalently if~$K=\red(K)$).  We denote by $\intro*\IrrAB$ the
"irreducible" functions in~$(B^*+1)^{A^*}$.

\AP What we have constructed is a bijection between
\begin{align*}
  \monadTrans(\IrrAB)=B^*\times\IrrAB +1&&\text{and}&&(B^*+1)^{A^*}\ ,
\end{align*}
that is defined as
\begin{align}
  \label{eq:1}
  \begin{array}{rl}
    \varphi\colon B^*\times\IrrAB +1 & \to (B^*+1)^{A^*} \\ 
    (u,K)&\mapsto u\pstar K\\
    \bot&\mapsto \botfun\ ,
  \end{array}
\end{align}
and the converse of which maps every $K\neq\botfun$ to
$(\lcpK,\redK)$, and $\botfun$ to~$\bot$.

Given $a\in A$ and $K\in (B^*+1)^{A^*}$ we denote by $a^{-1}K$ the
function in $(B^*+1)^{A^*}$ that maps $w\in A^*$ to $K(aw)$.

We can now define a functor $\AutFin(\langLKlT)\colon\catI\to
\KlTrans$ by setting
\begin{align*}
  &\AutFin(\langLKlT)(\objectLeft)=1&
  \AutFin(\langLKlT)(\objectCenter)=\IrrAB & &
  \AutFin(\langLKlT)(\objectRight)=1
\end{align*}
and defining $\AutFin(\langLKlT)$ on arrows as follows
\begin{align*}
  \AutFin(\langLKlT)(\symbLeft)&\colon 1\tokl \IrrAB  & 0&\mapsto (\lcpL, \redL)\\
  \AutFin(\langLKlT)(\symbRight)&\colon  \IrrAB \tokl 1 & K&\mapsto K(\varepsilon)\\
  \AutFin(\langLKlT)(a)&\colon \IrrAB \tokl \IrrAB &  K&\mapsto \botfun \qquad\qquad\qquad\qquad\text{if } a^{-1}K=\botfun\\
  &&K&\mapsto
  (\lcp(a^{-1}K),\red(a^{-1}K))\quad\text{otherwise}.
\end{align*}

\begin{lemma}
  The "$\KlTrans$-automaton" $\AutFin(\langLKlT)$ is a final object in
  $\catAutoLKlT$.
\end{lemma}

\begin{proof}
  We show that $\AutUTrans(\AutFin(\langLKlT))$ is isomorphic to the
  final automaton $\AutFin(\langLSet)$. Indeed, at the level of the
  state objects the bijection between
  $\AutUTrans(\AutFin(\langLKlT))(\objectCenter)$ and
  $\AutFin(\langLSet)(\objectCenter)$ is given by the function
  $\varphi$ defined in~\eqref{eq:1}. One can check that on arrows
  $\AutUTrans(\AutFin(\langLKlT))$ is the same as $\AutFin(\langLSet)$
  up to the correspondence given by $\varphi$.
\end{proof}

\subsection{A factorization system on $\catAutoLKlT$}
\label{subsection:factorization transducers}

\intro[factorization transducer]{}%
The factorization system on $\catAutoLKlT$ is obtained using
Lemma~\ref{lem:lifting-fact-syst} from a factorization system on
$\KlTrans$. There are several non-trivial factorization systems on
$\KlTrans$, one of which is obtained from the regular epi-mono
factorization system on $\Set$, or equivalently, from the regular
epi-mono factorization system on the category of Eilenberg-Moore
algebras for~$\monadTrans$. Notice that this is a specific result for
the monad $\monadTrans$ since in general, there is no reason that the
Eilenberg-Moore algebra obtained by factorizing a morphism between
free algebras be free itself. Nevertheless, in order to capture
precisely the syntactic transducer defined by
Choffrut~\cite{Choffrut79,Choffrut2003}, we will provide yet another
factorization system $(\EpiKlT,\MonoKlT)$, which we define concretely
as follows.  Given a morphism $f\colon X\tokl Y$ in $\KlTrans$ we
write $\pi_1(f)\colon X\to B^*+\{\bot\}$ and $\pi_2(f)\colon X\to
Y+\{\bot\}$ for the `projections' of $f$, defined by
\begin{align*}
\pi_1(f)(x)&=
\begin{cases}
u & \text{if }f(x)=(u,y)\,,\\
\bot & \text{otherwise,}
\end{cases}
&
\text{and}\qquad 
\pi_2(f)(x)&=
\begin{cases}
y & \text{if }f(x)=(u,y)\,,\\
\bot & \text{otherwise.}
\end{cases}
\end{align*}
We say that a partial function $g\colon X\to Y+\{\bot\}$ is surjective
when for every $y\in Y$ there exists $x\in X$ so that $g(x)=y$.

\AP The class $\intro{\EpiKlT}$ consists of all the morphisms of the
form $e\colon X\tokl Y$ such that $\pi_2(e)$ is surjective and the
class $\intro{\MonoKlT}$ consists of all the morphisms of the form
$m\colon X\tokl Y$ such that $\pi_2(m)$ is injective and $\pi_1(m)$ is
the constant function mapping every $x\in X$ to $\varepsilon$.

\begin{lemma}
  $(\EpiKlT, \MonoKlT)$ is a factorization system on $\KlTrans$.
\end{lemma}

\begin{proof}
  Notice that $f$ is an isomorphism in $\KlTrans$ if and only if
  $f\in\EpiKlT\cap\MonoKlT$.

  If $f\colon X\tokl Y$ is a morphism in $\KlTrans$ then we can define
  \[Z=\{y\in Y\mid \exists x\in X\ldotp\exists u\in B^*\ldotp  
  f(x)=(u,y)\}\,.\] We define $e\colon X\tokl Z$ by $e(x)=f(x)$ and
  $m\colon Z\tokl Y$ by $m(y)=(\epsilon,y)$. One can easily check that
  $f=m\circ e$ in $\KlTrans$.

  Lastly, we can show that the "diagonal property" holds. Assume we
  have a commuting square in $\KlTrans$.
\[
\begin{tikzcd}
  X\ar[r, two heads, "e"]\ar[d,swap, "f"] & Y\ar[d, "g"]\ar[dl, dashed,"d"] \\
  Z\ar[r, tail,swap, "m"] & W
\end{tikzcd}
\]
We will prove the existence of $d\colon Y\tokl Z$ so that $d\circ e=f$
and $m\circ d=g$. Assume $y\in Y$. If $g(y)=\bot$ we set
$d(y)=\bot$. Otherwise assume $g(y)=(v,t)$, for some $v\in B^*$ and
$t\in W$. Since $e\in \EpiKlT$, there exists $u\in B^*$ and $x\in X$
so that $e(x)=(u,y)$. Assume $f(x)=(w,z)$ for some $w\in B^*$ and
$z\in Z$. We set $d(y)=(v, z)$. First, we have to prove that
this definition does not depend on the choice of $x$.

Assume that we have another $x'\in X$ so that $g(x')=(u',y)$ and
assume $f(x')=(w',z')$. Using the fact that $m\in \MonoKlT$, we will
show that $z=z'$, and thus $d(y)$ is well defined. Indeed, notice that

\begin{tabular}{lll}
$\begin{cases}
  g\circ e(x)=(uv,t)\\
  g\circ e(x')=(u'v,t)\,,\\
\end{cases}
$
&
 or equivalently, 
&
$
\begin{cases}
  m\circ f(x)=(uv,t)\\
  m\circ f(x')=(u'v,t)\,,\\
\end{cases}
$
\end{tabular}

\noindent Assume that $m(z)=(\varepsilon, t_1)$ and
$m(z')=(\varepsilon, t_2)$. This entails
\[
\begin{cases}
  m\circ f(x)=(uv,t)=(w,t_1)\\
  m\circ f(x')=(u'v,t)=(w',t_2)\,.\\
\end{cases}
\]
We obtain that $t_1=t_2=t$. Since $m\in\MonoKlT$ (and thus $\pi_2(m)$
is injective) we get that $z=z'$, which is what we wanted to prove.
It is easy to verify that $d\circ e=f$ and $m\circ d=g$.
\end{proof}

This completes the proof of Theorem~\ref{theorem:minimization
  transducer}.

\section{Brzozowski's minimization algorithm}
\label{sec:brzozowski}

\newrobustcmd\determinize{\kl[\determinize]{\mathtt{determinize}}}
\newrobustcmd\codeterminize{\kl[\codeterminize]{\mathtt{codeterminize}}}
\newrobustcmd\transpose{\kl[\transpose]{\mathtt{transpose}}}

\subsection{Presentation}

""Brzozowski's algorithm"" is a minimization algorithm for
automata. It takes as input a "non-deterministic@non-deterministic
automaton" "automaton@non-deterministic automaton"~$\mathcal{A}$, and
computes the deterministic automaton:
\begin{align*}
  \determinize(\transpose(\determinize(\transpose(\mathcal A)))),
\end{align*}
in which
\begin{itemize}
\item\AP $\intro*\determinize$ is the operation from classical
  automata theory that takes as input a deterministic automaton,
  performs the powerset construction to it and at the same time
  restricts to the reachable states, yielding a deterministic
  automaton, and
\item\AP $\intro*\transpose$ is the operation that takes as input a
  non-deterministic automaton, reverses all its edges, and swaps the
  role of initial and final states (it accepts the mirrored language).
\end{itemize}
\AP In this section, we will establish the correctness of Brzozowski's
algorithm: this sequence of operations yields the minimal automaton
for the language. For easing the presentation we shall present the
algorithm in the form:
\begin{align*}
  \determinize(\codeterminize(\mathcal A)),
\end{align*}
in which $\intro*\codeterminize$ is the operation that takes a
non-deterministic automaton, and constructs a backward deterministic
one (it is equivalent to the sequence
$\transpose\circ\determinize\circ\transpose$).

In the next section, we show how $\determinize$ and~$\codeterminize$
can be seen as adjunctions, and we use it immediately after in a
correctness proof of "Brzozowski's algorithm".

\subsection{Non-deterministic automata and determinization}
\label{sec:det-codet-adj}

\noindent
\begin{minipage}{12cm}
  A ""non-deterministic automaton"" is completely determined by the
  relations described in the following diagram, where we see the
  initial states as a relation from $1$ to the set of states $Q$, the
  final states as a relation from $Q$ to $1$ and the transition
  relation by any input letter $a$, as a relation on the set $Q$.
\end{minipage}
~~
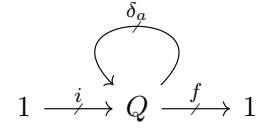
\begin{tikzcd}
  1\arrow[r,negated,"i"] & Q\arrow[loop,looseness=6,negated,swap,
  "\delta_a"]\arrow[r,negated,"f"] & 1
\end{tikzcd}
\medskip

We can model nondeterministic automata as functors by taking as output
category $\Rel$ -- the category whose objects are sets and maps are
relations between them. We consider "$\Rel$-automata" $\autA\colon
\catI\to\Rel$ such that $\autA(\objectLeft)=1$ and
$\autA(\objectRight)=1$. In this section we show how to determinize a
"$\Rel$-automaton", that is, how to turn it into a "$\Set$-automaton"
and how to codeterminize it, that is, how to obtain a
"$\Setop$-automaton", all recognizing the same language.

Given a language $L\subseteq A^*$ we can model it in several
equivalent ways: as a "$(\Set,1,2)$-language" $\langLSet$, or as a
"$(\Setop,2,1)$-language" $\langLSetop$, or, lastly as a
"$(\Rel,1,1)$-language" $\langLRel$. This is because we can model the
fact $w\in L$ using a morphisms in either of the three isomorphic
hom-sets
\[
\Set(1,2)\cong \Setop(2,1)\cong \Rel(1,1)\,.
\]
Determinization and codeterminization (without assuming the
restriction to reachable states as in $\determinize$ and
$\codeterminize$) of a $\Rel$-automaton can be seen as applications of
Lemma~\ref{lem:lifting-adjunctions} and are obtained by lifting the
adjunctions between $\Set$, $\Rel$ and $\Setop$.
\begin{equation}
  \label{eq:2}
  \begin{tikzcd}[column sep={0.9mm},row sep={3.5mm}]
    \catAutoLSet\ar[ddd,""] \ar[rr,bend left=15, "\overline{\FPow}"] &
    \bot &\catAutoLRel\ar[ll,bend left=15,
    "\overline{\UPow}"]\ar[ddd,""]
    \ar[rr,bend left=15, "\overline{\UPowop}"] & \bot & \catAutoLSetop \ar[ll,bend left=15, "\overline{\FPowop}"]\ar[ddd," "]\\
    & & \\
    & & \\
    \Set \ar[rr,bend left=20, "\FPow"] & \bot &\Rel\ar[ll,bend
    left=20, "\UPow"] \ar[rr,bend left=20, "\UPowop"]& \bot & \Setop
    \ar[ll,bend left=20, "\FPowop "]
  \end{tikzcd}
\end{equation}
\AP\intro[\FPowop]{}%
 The adjunction between $\Set$ and $\Rel$ is the Kleisli adjunction
for the powerset monad: $\intro*\FPow$ is identity on objects and maps
a function $f\colon X\to Y$ to itself $f\colon X\tokl Y$, but seen as
a relation.
\AP\intro[\UPowop]{}%
The functor $\intro*\UPow$ maps $X$ to its powerset
$\Pow(X)$, and a relation $R\colon X\to Y$ to the function
$\UPow(R)\colon\Pow(X)\to\Pow(Y)$ mapping $A\subseteq X$ to $\{y\in
Y\mid\exists x\in A. (x,y)\in R\}$.

\phantomintro\UPowop\phantomintro\FPowop The adjunction between
$\Setop$ and $\Rel$ is the dual of the previous one, composed with the
self-duality of $\Rel$.  The left adjoint $\overline\FPow$ transforms
a deterministic automaton into a non-deterministic one, while the
right adjoint $\overline{\UPow}$ is the ""determinization
functor"". On the other hand, the left adjoint functor
$\overline{\UPowop}$ is the ""codeterminization functor"".

\subsection{Brzozowski's minimization algorithm}
\label{sec:duality}

The correctness of "Brzozowski's algorithm" can be seen in the
following chain of adjunctions from Lemma~\ref{lem:adj-reach-obs} and
\eqref{eq:2} (that all correspond to equivalences at the level of
languages):
\[
\begin{tikzcd}[column sep={14mm,between origins}]
  \catReachAutoLSet \ar[rr,bend left, hook,"E"] &\bot
  &\catAutoLSet\ar[rr,bend left, "\overline{\FPow}"]\ar[ll,bend left,
  "\Reach"] & \bot &\catAutoLRel\ar[ll,bend left, "\overline{\UPow}"]
  \ar[rr,bend left, "\overline{\UPowop}"] & \bot & \catAutoLSetop
  \ar[ll,bend left, "\overline{\FPowop}"]\ar[rr,bend left, "\Obs"]
  &\bot~ & \catObsAutoLSetop\ar[ll,bend left, hook,"E"]
\end{tikzcd}
\]

A path in this diagram corresponds to a sequence of transformations of
automata. It happens that when~$\Obs$ is taken, the resulting
automaton is "observable", i.e., there is an injection from it to the
"final object". This property is preserved under the sequence of right
adjoints $\Reach\circ\overline{\UPow}\circ 
\overline{\FPowop}\circ E$. Furthermore, after application
of~$\Reach$, the automaton is also "reachable". This means that
applying the sequence
$\Reach\circ\overline{\UPow}\circ\overline{\FPowop}\circ
E\circ\Obs\circ\overline{\UPowop}$ 
to a non-deterministic automaton produces a deterministic and minimal
one for the same language.  We check for concluding that the sequence
$\Obs\circ\overline{\UPowop}$ is what is implemented by
$\codeterminize$, that the composite $\overline{\FPowop}\circ E$
essentially transforms a backward deterministic "observable" automaton
into a non-deterministic one, and that finally
$\Reach\circ\overline{\UPow}$ is what is implemented by
$\determinize$. Hence, this indeed is "Brzozowski's algorithm".

\begin{remark}
  The composite of the two adjunctions in~\eqref{eq:2} is almost the
  adjunction of~\cite[Corollary~9.2]{BonchiBHPRS14} upon noticing that
  the category $\catAutoLSetop$ of $\Set^\op$-automata accepting a
  language $\langLSetop$ is isomorphic to the opposite of the category
  $\catAutoLSetrev$ of $\Set$-automata that accept the reversed
  language seen as functor $\langLSetop$. This observation in turn can
  be proved using the symmetry of the input category $\catI$.
\end{remark}

\section{Conclusion}
\label{sec:conclusion}
In this paper we propose a view of automata as functors and we showed
how to recast well understood classical constructions in this setting,
and in particular minimization of subsequential transducers. The
applications provided here are a small sample of many possible further
extensions. We argue that this perspective gives a unified view of
language recognition and syntactic objects. We can change the input
category $\catI$, so that we obtain monoids instead of automata, or
more generally, other algebras as recognisers for
languages. Minimization works out following the same recipe and yields
the syntactic monoid (algebra) of a language. We can go beyond regular
languages and obtain in this fashion the ``syntactic space with an
internal monoid'' of a possibly non-regular
language~\cite{GehrkePR16}. Our functorial treatment of automata is
more general than that presented in~\cite{Ballester-Bolinches15} and
it would be interesting to explore equations and coequations in this
setting. We hope we can extend the framework to work with tree
automata in monoidal categories. We discussed mostly NFA
determinization, but we can obtain a variation of the generalized
powerset construction~\cite{SilvaEtAl:genPow} in this framework.

\end{document}